\documentclass[journal]{IEEEtran} 		
\usepackage{amsmath,amssymb}
\usepackage{amsthm}
\usepackage{color}
\usepackage[mathscr]{eucal}
\usepackage{graphics,graphicx,multicol}
\usepackage{epsfig}
\usepackage{enumerate}
\usepackage{subfigure}
\usepackage{algorithmic}
\usepackage[section]{algorithm}
\usepackage{morefloats}
\usepackage{multirow}
\usepackage{array}
\usepackage{pstricks, pst-node, pst-plot, pst-circ}
\usepackage{moredefs}
\usepackage{courier}
\usepackage{colortbl}
\usepackage{xcolor}
\usepackage{cite}

\newtheorem{thm}{Theorem}[section] 
\newtheorem{lem}[thm]{Lemma}

\begin{document}
\title{A Multi-Tier Wireless Spectrum Sharing System Leveraging Secure Spectrum Auctions}
\author{Ahmed~Abdelhadi,
        Haya~Shajaiah,
        and~Charles~Clancy,
\thanks{A. Abdelhadi, H. Shajaiah, and C. Clancy are with the Hume Center for National Security and Technology, Virginia Tech, Arlington,
VA, 22203 USA e-mail: \{aabdelhadi, hayajs, tcc\}@vt.edu.}
}

\maketitle
\begin{abstract}
Secure spectrum auctions can revolutionize the spectrum utilization of cellular networks and satisfy the ever increasing demand for resources. In this paper, a multi-tier dynamic spectrum sharing system is studied for efficient sharing of spectrum with commercial wireless system providers (WSPs), with an emphasis on federal spectrum sharing. The proposed spectrum sharing system optimizes usage of spectrum resources, manages intra-WSP and inter-WSP interference and provides essential level of security, privacy, and obfuscation to enable the most efficient and reliable usage of the shared spectrum.  It features an intermediate spectrum auctioneer responsible for allocating resources to commercial WSPs by running secure spectrum auctions. The proposed secure spectrum auction, \textbf{MTSSA}, leverages Paillier cryptosystem to avoid possible fraud and bid-rigging. Numerical simulations are provided to compare the performance of MTSSA, in the considered spectrum sharing system, with other spectrum auction 
mechanisms for realistic cellular systems.
\end{abstract}
\providelength{\AxesLineWidth}       \setlength{\AxesLineWidth}{0.5pt}%
\providelength{\plotwidth}           \setlength{\plotwidth}{8cm}
\providelength{\LineWidth}           \setlength{\LineWidth}{0.7pt}%
\providelength{\MarkerSize}          \setlength{\MarkerSize}{3pt}%
\newrgbcolor{GridColor}{0.8 0.8 0.8}%
\newrgbcolor{GridColor2}{0.5 0.5 0.5}%
\section{Introduction}\label{sec:intro}
Traditionally, radio spectrum management is controlled by a central government agency such as the Federal Communications Commission (FCC) in the United States. Such a centralized spectrum assignment mechanism predetermines static bands for specific usage without taking into consideration the service requirements and the dynamic nature of the radio spectrum. This results in an under-utilized pre-assigned spectrum bands while many of the commercial bands are overcrowded due to the rapid growth of wireless services. To address this limitation in the spectrum utilization, the FCC has legalized secondary markets for spectrum such that a primary spectrum licensee can lease its under-utilized spectrum to secondary incumbents \cite{FCC-2004}. Inspired by microeconomics mechanisms \cite{Economic-Framework,Fair-Profit,Stackelberg-Game}, spectrum auction seems to be a promising solution to release the under-utilized spectrum to potential secondary users \cite{Zhou:2008:ESS,Trust,Jia:2009:RGT}.
There has been some previous work to deal with security issues in auction design. These works have focused on adding some new features to the auction design, such as confidentiality, fairness \cite{New-Sealed-Bid,Batch-Verification} and anonymity.

Because of the reusability feature of the radio spectrum, traditional auctions can not be directly used in a spectrum auction design. Spectrum auctions should allow bidders, that are not within the interference radius of each other, to use the same frequency simultaneously. Therefore, the optimal spectrum allocation is considered NP-complete \cite{Framework-for-Wireless,Jain:2003} whereas conventional auctions are based on optimal allocations \cite{Zhou:2008:ESS}. In addition, a spectrum auction design is challenged by the effect of the back-room dealing, between insincere bidders and the auctioneer, to the whole network. A secure spectrum auction needs to avoid possible frauds of the auctioneer and bid-rigging between the bidders and the auctioneer.

In this paper, we design a secure spectrum auction that leverages Paillier cryptosystem, \textbf{MTSSA}, to avoid possible frauds and bid-riggings and provide a framework for a multi-tier spectrum sharing system to achieve an efficient utilization for the under-utilized spectrum.

\subsection{Related Work}\label{relatedwork}
Most early works in spectrum auctions, such as \cite{Jia:2009:RGT, Zhou:2008:ESS}, have focused on single-seller multi-buyer auctions that deal with homogeneous channels. In \cite{Zhou:2008:ESS}, the authors have proposed \textbf{VERITAS}, a truthful mechanism that supports an eBay-like dynamic spectrum market. It is a good fit for short term and small regions based spectrum auction which is not the case in FCC required spectrum auction which is for long term and large geographical regions. To deal with interference between neighboring bidders, a conflict graph and a wireless spectrum auction framework have been proposed in \cite{Zhou:2008:ESS}. Based on these concepts, a conflict graph is used to represent the interference relationship in \textbf{VERITAS} \cite{Zhou:2008:ESS}. In a sealed secondary price and VCG auctions, the dominant strategy for certain bidder, when he has no information about other bidders' bids, is to bid with his true evaluation values \cite{AuctionTheory}. The authors in \cite{multi-
winner} have showed that it is not always right to allocate spectrum bands to the bidder with the highest bid, as proposed in \cite{Zhou:2008:ESS}, if the sum of the neighbors bids is much higher than the highest bid. Their proposed solution is based on grouping nodes such that nodes with no interference are grouped together. However, their group partition approach is NP-complete under interference constraints \cite{Jain:2003}.

The authors in \cite{Trust} have proposed \textbf{TRUST}, a spectrum trading approach that satisfies some good properties. However, it achieves truthfulness while sacrificing one group of bidders, as it takes the group's bid as the clearing price. In \cite{Strategy-Proof}, the authors have improved the idea of \textbf{TRUST} as they succeeded to achieve truthfulness by only sacrificing one buyer in each group. But, both works \cite{Trust, Strategy-Proof} have inherited \textbf{McAfee} mechanism \cite{McAfee} which requires homogeneous channels. In \cite{Yang:2011:TAC}, the proposed \textbf{TASC} mechanism was the first to consider heterogeneous channels. However, it can reduce the system efficiency as all channels are restricted to a unique clearing price. In \cite{TAHES,TAMES}, \textbf{TASC} mechanism has been extended to consider spectrum reusability and diversity of channel characteristics. In \cite{SPRING}, the authors have proposed a privacy preserving auction for spectrum trading. In \cite{iDEAL, Win-
Coupon}, an auction based framework is purposed. A third party leases its unused resources to service providers to provide dynamic cellular offloading.

In \cite{Zhou:2008:ESS,TAMES}, the authors have exploited frequency interference property. They used interference graph model that makes spectrum allocation, allows spectrum reuse and avoids interference. In \cite{Trust,Strategy-Proof}, the authors have utilized the reusability property by dividing buyers into groups such that buyers in the same group do not interfere with each other. Each group of buyers either wins or loses the same channel.

Most of existing works have failed to consider spectrum bands as non identical bands. Spectrum reusability in an auction design has been first addressed in \cite{Trust}. In \cite{Combinatorial-auction}, the authors have modeled a spectrum auction based on spectrum reusability in a time-frequency division manner. The authors in \cite{SMALL} have also considered spectrum reusability in their auction design by assuming that each spectrum buyer is allowed to have multiple radios. The proposed \textbf{MTSSA} scheme also supports the frequency reusability property. Moreover, \textbf{MTSSA} supports the case of heterogeneous frequency bands.

Beside the properties of secondary price auctions that are beneficial to have in a spectrum auction, i.e. such as incentive compatibility, individual rationality and no positive transfers, it is important to secure the spectrum auction to avoid potential back room dealing. An ideal spectrum auction design would allow the auctioneer to find the best allocation of the frequency bands, determine the winners and their payments while the bidders keep their actual bidding values secret and unknown to the auctioneer. This can prevent frauds made by insincere auctioneers and bid rigging between the auctioneer and the bidders. There has been some previous works in secure spectrum auctions. The authors in \cite{Secure-Multi-agent,Secure-Generalized-2003,Secure-Generalized-2004}, have used homomorphic encryption to secure traditional auction designs. In \cite{Themis}, the authors have considered frequency reuse in their secure spectrum auction design, and propose \textbf{THEMIS}. However, \textbf{THEMIS} does not 
support multi-tier spectrum sharing systems where spectrum reuse is possible among multiple service providers. In these systems a dynamic spectrum sharing approach is required to provide an efficient sharing of the spectrum among multiple service providers.

In this paper we design a truthful secure spectrum auction framework by considering the spectrum spatial reuse property and the heterogeneous propagation properties of different frequency bands. We propose \textbf{MTSSA}, a secure spectrum auction design that provides framework for a multi-tier dynamic spectrum sharing system and optimizes allocating the spectrum resources that are managed by a broker (i.e. the auctioneer). It allows the auctioneer to allocate its under-utilized frequency bands, leased from federal government, to commercial WSPs by running secure spectrum auction. By leveraging Paillier cryptosystem, \textbf{MTSSA} can prevent possible frauds and bid-rigging.

\subsection{Our Contributions}\label{sec:contributions}
The major contributions of the proposed spectrum auction are summarized as:
\begin{itemize}
\item \textbf{MTSSA} considers spectrum reusability and the case of heterogeneous frequency bands, e.g. commercial and federal bands.
\item \textbf{MTSSA} provides a framework for a multi-tier dynamic spectrum sharing system that allows an efficient spectrum sharing of the under-utilized spectrum with commercial WSPs. The auctioneer allocates the under-utilized frequency bands to commercial WSPs' BSs by running a secure spectrum auction. \textbf{MTSSA} optimizes the usage of spectrum resources by managing intra-WSP and inter-WSP interference. In order to account for frequency reusability, the auctioneer divides the network into subnets and auctions the frequency bands in each of the subnets one after another. Each bidder (i.e. BS), maintains a conflict-table. The bidder updates the conflict-table and share his spectrum occupancy status with his neighbors when changes are made.
\item \textbf{MTSSA} provides a truthful auction that is achieved when each bidder submits its true evaluation value. Truthfulness is a dominant strategy for \textbf{MTSSA} as it prevents manipulating the auction.
\item \textbf{MTSSA} uses a payment method that satisfies some essential economic properties such as incentive compatibility, individual rationality and no positive transfers.
\item \textbf{MTSSA} leverages Paillier cryptosystem \cite{Paillier:1999:PKC, Paillier:1999:EPC, Themis} to create a ciphertext for the bidding values. Each BS submits its bidding values through a buffer that creates an encrypted version of the bidding values. While the actual bidding values are kept secret from the auctioneer, the auctioneer is still able to reveal the auction results and charge the bidders securely.
\item \textbf{MTSSA} provides a secure spectrum auction that prevents frauds of insincere auctioneers and bid-rigging while achieving an efficient spectrum utilization, revenue and bidders' satisfaction.
\end{itemize}

The remainder of this paper is organized as follows. In Section \ref{sec:SystemModel}, the spectrum trading architecture is described and the system model for the Multi-Tier spectrum sharing secure auction MTSSA is outlined. In Section \ref{sec:Considerations}, design considerations are presented. We describe the payment method for the proposed auction in Section \ref{sec:Payment}, prove its satisfaction of some desired economic properties in Section \ref{sec:Properties} and describe the design challenges in \ref{DesignChallenges}. In Section \ref{sec:MTSSA}, we present the frequency bands allocation procedure and the encryption design using Paillier cryptosystem for the proposed MTSSA. Simulation set up and performance analysis are discussed in  Section \ref{sec:Simulation}. Section \ref{sec:conclude} concludes the paper.

\section{System Model}\label{sec:SystemModel}

\subsection{Spectrum Trading Architecture}\label{sec:Architecture}
We consider a spectrum trading scenario where the spectrum owner is a federal regulatory agency that leases its under-utilized spectrum on a long-term basis to a broker which manages spectrum assets and plays the role of a middleman for the spectrum owner of the under-utilized spectrum, e.g. federal government, and the WSPs. The architecture of this spectrum assignments is represented through a spectrum pyramid as shown in Figure \ref{fig:pyramid2pdf}. At the top of this pyramid, is the spectrum owner that leases the under-utilized frequency bands to a spectrum broker under certain rules \cite{FCC-2003,FCC-2004,FCC-2010}.
The broker represents a secondary market place that auctions these frequency bands to WSPs. At the bottom of  the pyramid, are the end users devices (i.e. users equipments (UEs)) that are assigned spectrum by the WSPs base stations (BSs). In this paper we focus on designing a secure spectrum auction between the broker (i.e. the auctioneer) and the WSPs base stations to allocate the under-utilized frequency bands.
\begin{figure} [h]
\centering
\includegraphics[height=2in, width=3.5in]{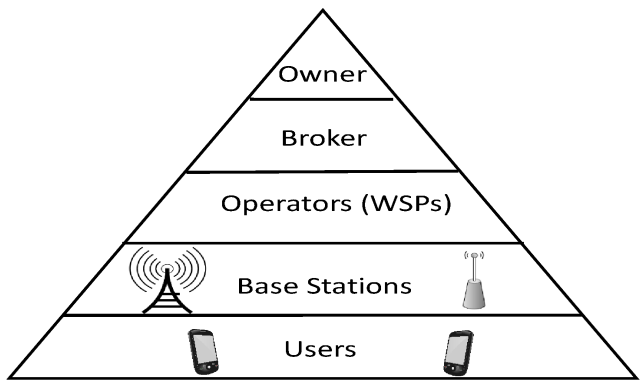}
\caption{A spectrum pyramid that represents an architecture for the under-utilized spectrum assignments.}
\label{fig:pyramid2pdf}
\end{figure}

\subsection{Spectrum Auction Model}\label{sec:Auction Model}
Consider a spectrum auction setting, where one auctioneer (i.e. the broker in Figure \ref{fig:pyramid2pdf}) auctions a set of frequency bands $\mathcal{M}=\{1,2,...,M\}$ to $\mathcal{N}=\{1,2,...,N\}$ bidders (i.e. nodes representing BSs) located in the same geographical region where $\mathcal{N}$ represents a set of all bidders that belong to different WSPs. Let $L$ be the number of WSPs where each WSP has a coverage area within the auction's geographical region. Each WSP (i.e. the $l^{th}$ WSP) provides a mobile wireless service over multiple cellular cells. Its cellular network consists of macro cells and small cells. Within the coverage area of some macro cells, there exist one or more small cells with pico/femto BSs, see Figure \ref{fig:WSPs}. Let $\mathcal{N}^l$ be the set of all of macro cells and small cells BSs that belong to the $l^{th}$ WSP and let $\mathcal{N}$ be a set of all nodes that belong to the $L$ WSPs where $\mathcal{N}=\mathcal{N}^1\cup \mathcal{N}^2\cup... \cup \mathcal{N}^L$ and $N=|\mathcal{N}|$.

In this paper, we consider a multi-band spectrum auction where each node can bid for a single or multiple frequency bands from the set of available frequency bands $\mathcal{M}$  based on its demand. Once the broker leases the spectrum owner's unused frequency bands in $\mathcal{M}$ for a time duration $T$, the broker becomes the owner of the spectrum bands in $\mathcal{M}$. Meanwhile, the interested WSPs submit their bids to the broker. Let the allocation set $\mathcal{K}=\{\alpha^1, \alpha^2,...\}$ denotes the set of all possible allocations of the frequency bands $\mathcal{M}$. For example, given that $\mathcal{M}=\{1,2\}$ and $\mathcal{N}=\{1,2\}$, we have $\mathcal{K}=\{\alpha^1=(\{1,2\},\{\}),\alpha^2=(\{1\},\{2\}),\alpha^3=(\{2\},\{1\}),\alpha^4=(\{\},\{1,2\})\}$, where $(\alpha^1=(\{1,2\},\{\})$ denotes that frequency bands $1$ and $2$ are allocated to bidder $1$ and nothing to bidder $2$. Each node submits its sealed bids $\mathbf{b_n}=[b_n(\alpha^1),b_n(\alpha^2),...]$, e.g.  $\mathbf{b_1}
=[2,1,1,0]$ indicates that node $1$ bids $2$ for allocation $\alpha^1$, $1$ for allocation $\alpha^2$, $1$ for allocation $\alpha^3$ and $0$ for allocation $\alpha^4$. For certain allocation $\alpha=\alpha_n$, each node $n$ has a true evaluation value $v_n(\alpha)$. Let $\mathbf{v_n}=[v_n(\alpha^1),v_n(\alpha^2),...]$ be the true evaluation vector for node $n$. Let $p_n$ represents the price that is charged by the auctioneer to bidder $n$ for allocating the frequency bands. The utility of bidder $n$, denoted by $U_n$, is defined as the difference between the bidder's true evaluation value and the actual price it pays to the auctioneer $p_i$, $U_n=v_n(\alpha)-p_n$, for a specific allocation $\alpha$. The Auctioneer's revenue from the spectrum sales is defined as $R=\sum_{n=1}^{n=N}p_n$. We assume that bidders can submit different bids for different combinations of the frequency bands. Table \ref{table:notations} summarizes some of the notations used in the design.
\begin{table}[ht]
\caption{Key symbols in this paper}
\centering
\begin{tabular}{l p{7cm}}
\hline \\
$\mathcal{M}$ & Frequency bands set \\ 
$\mathcal{N}^l$ & Bidders set for all bidders that belong to WSP $l$ \\
$\mathcal{N}$ & Bidders set of all nodes that belong to the $L$ WSPs, $\mathcal{N}=\mathcal{N}^1\cup \mathcal{N}^2\cup... \cup \mathcal{N}^L$ \\
$\mathcal{K}$ & Allocation set $\mathcal{K}=\{\alpha^1,\alpha^2,...,\alpha^i\}$ \\
$\mathbf{b_n}$ & Node $n$ sealed bids vector for the allocation set $\mathcal{K}$ \\
$\mathbf{v_n}$ & True evaluation vector of BS $n$ \\
$p_n$ & Price charged by the auctioneer to BS $n$ \\
$U_n$ & Bidder's utility, $U_n=v_n(\alpha)-p_n$  \\
$R$ & Auctioneer's revenue, $R=\sum_{n=1}^{n=N}p_n$ \\  \\
\hline 
\end{tabular}
\label{table:notations} 
\end{table}

In Figure \ref{fig:WSPs}, we show two WSPs (i.e. $L=2$) providing service in the same geographical region where the broker performs its spectrum auction. Both WSPs are interested in the auctioneer's frequency bands $\mathcal{M}$. Therefore, both of them participate in the spectrum auction. In the coverage area of each WSP there exists multiple macro cells and small cells managed by that WSP. BSs requesting additional frequency bands submit sealed bidding vectors to the auctioneer via an intermediate secure gateway to participate in the auction of the under-utilized federal spectrum bands. Considering the frequency reuse property \cite{Jain:2003,Trust}, each BS has certain coverage radius (i.e. assume it is equivalent to the cell's radius). Within the coverage radius of the $n^{th}$ BS, non of the interfering BSs can simultaneously use any of the frequency bands that the $n^{th}$ BS is using. However, a non-interfering BS can use the same frequency band that is simultaneously used by a BS located outside its 
coverage radius without causing interference, i.e. frequency reuse is utilized in our model. The auctioneer constructs an interference conflict graph for all the BSs that are participating in the auction.
\begin{figure}
\centering
\includegraphics[height=2in, width=3.5in]{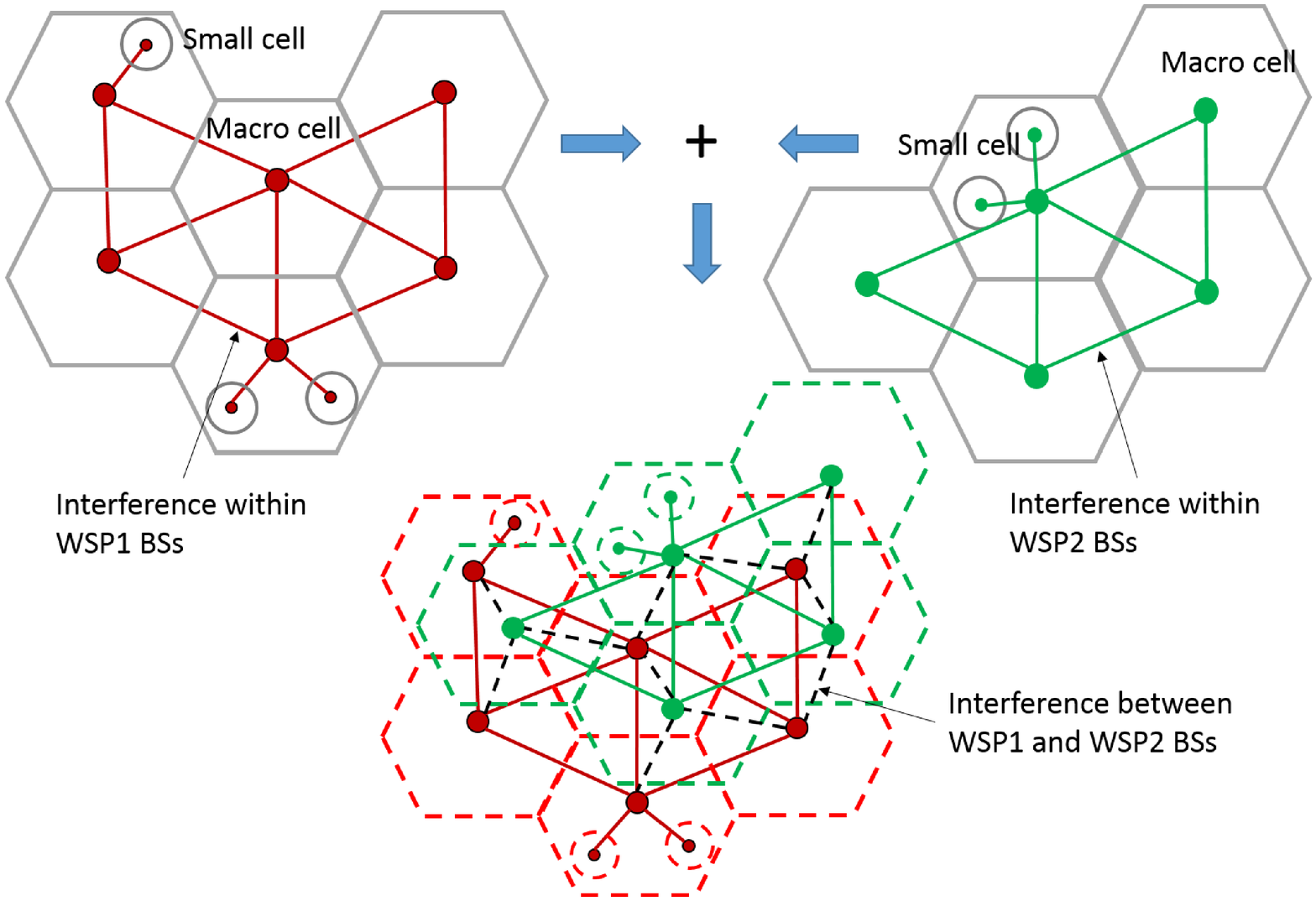}
\caption{Two WSPs with a coverage area within the geographical region where the auction takes place. In each WSP's macro cells and small cells, all the BSs that are interested in the auctioneer's under-utilized frequency bands are part of the interference conflict graph.}
\label{fig:WSPs}
\end{figure}
In Figure \ref{fig:subnets}, we show the frequency conflict graph with all bidders/BSs that belong to the two WSPs, each BS is connected with other BSs located within its coverage radius (i.e. bidder $n$ is connected with all BSs that must not simultaneously use same frequency bands due to interference between them) where the edges represent mutual interference between the corresponding BSs. The interference conflict graph can be constructed using physical or protocol channel model \cite{ahmed_infocom}. It is assumed that there exists a pilot channel, like the one in \cite{Haya_Utility3}, to exchange information between the auctioneer and the BSs or simply by sending that unsecured information with the bids.
\begin{figure}
\centering
\includegraphics[height=2in, width=3in]{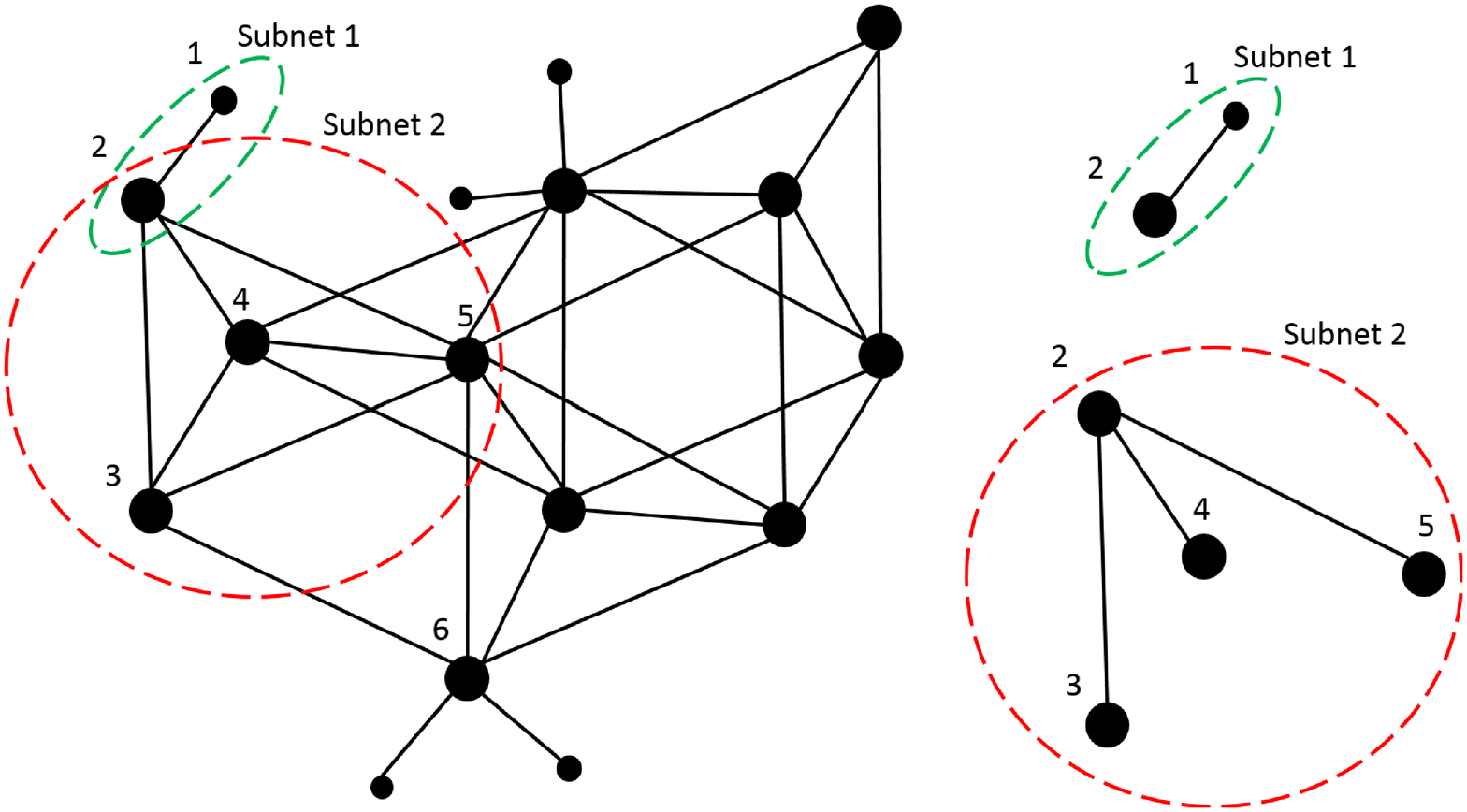}
\caption{Frequency conflict graph for all BSs that belong to the two WSPs shown in Figure \ref{fig:WSPs}. Each node represents one BS and the edges represent mutual interference between the end points (i.e. BSs). Subnet $1$ consists of the small cell's BS (i.e. BS $1$), which represents the root BS for the subnet, and the macro cell's BS (i.e. BS $2$). Subnet $2$ consists of BSs $2$, $3$, $4$ and $5$ where BS $2$ is the root BS.}
\label{fig:subnets}
\end{figure}
Furthermore, the proposed spectrum auction is executed in one subnet after another where a subnet is defined to be a group of BSs that includes one root BS, i.e. BS $n$, and all other BSs that are connected to it through interference edges (i.e. the BSs that have mutual interference with BS $n$) but not previously considered root BSs. Figure \ref{fig:subnets} shows two subnets in the frequency conflict graph of the two WSPs. In Figure \ref{fig:Three-Stage-Auction}, we show the spectrum auction model for \textbf{MTSSA} for two WSPs' participate in the spectrum auction. First, all BSs submit their encrypted bidding vectors via an intermediate gateway that is operated by federal government. The auctioneer then carries out a secure spectrum auction in one subnet after another. It then allocates the winning BSs frequency bands and charges them for the allocated resources.
\begin{figure}
\centering
\includegraphics[height=3.5in]{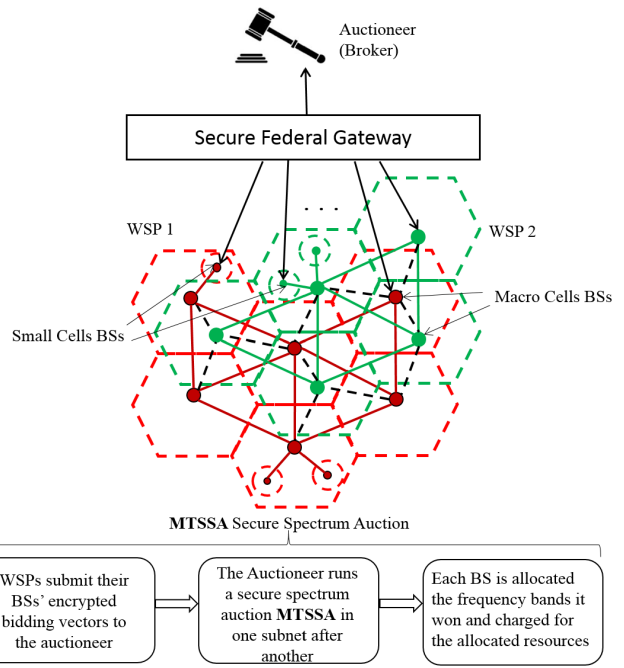}
\caption{Spectrum auction model for the proposed \textbf{MTSSA} with two WSPs' BSs participating in the auction.}
\label{fig:Three-Stage-Auction}
\end{figure}

\section{Design Considerations}\label{sec:Considerations}
In this section, we present the payment method for the proposed auction. We also discuss some economic properties that need to be considered in the design and prove that by using a VCG based auction approach some desired economic properties can be satisfied.
\subsection{The Payment Method}\label{sec:Payment}
Our goal is to use a payment rule that satisfies some of the required economic properties, such as incentive compatibility, individual rationality and no positive transfers. In addition, it is important for the payment rule to provide a satisfactory revenue for the auctioneer. Under certain assumptions, it has been proven that VCG auction satisfies these three economic properties while maximizing the auctioneer's revenue \cite{Incentives}. VCG auction is also proven to be Pareto efficient \cite{Vic61}. In VCG, each bidder submits its true evaluation values regardless of the bidding values that other bidders submit. This is a dominant strategy for the bidder to maximize its utility and win the auction. In our design, we use a payment method that is based on VCG mechanism with Clarke pivot payments \cite{Nisan2007:AGT}. Using this payment rule, each bidder pays the difference between the social welfare with and without his participation (i.e. bidder $n$ pays the externality he causes). Consider the same system 
setup as described in Section \ref{sec:SystemModel}. Each bidder $n$ submits its sealed bidding vector $\mathbf{b_n}$ for the allocation set $\mathcal{K}$. The auctioneer selects a Pareto efficient allocation $\alpha^*\in\mathcal{K}$ where $\alpha^*$ is defined as
\begin{equation} \label{eq:pareto1}
\alpha^* =\arg\max_ {\alpha\in\mathcal{K}}\sum_{n}b_n(\alpha).
\end{equation}
With truthful bidding values, the auctioneer assigns its frequency bands $\mathcal{M}$ based on $\alpha^* $ allocation. Furthermore, let $\alpha^*_{-n}\in\mathcal{K}$ be an allocation without node $n$ participating that is defined as
\begin{equation} \label{eq:pareto2}
\alpha^*_{-n} =\arg\max_ {\alpha\in\mathcal{K}}\sum_{k\neq n}b_k(\alpha).
\end{equation}
The auctioneer charges bidder $n$ a payment $p_n$ that is equivalent to
\begin{equation} \label{eq:payment}
p_n =\sum_{k\neq n}b_k(\alpha^*_{-n})-\sum_{k\neq n}b_k(\alpha^*).
\end{equation}
Then, the utility of bidder $n$ can be expressed as
\begin{equation} \label{eq:utility}
\begin{aligned}
U_n =
& \;v_n(\alpha^*)-p_n \\
=
& \;v_n(\alpha^*)-(\sum_{k\neq n}b_k(\alpha^*_{-n})-\sum_{k\neq n}b_k(\alpha^*))\\
=
& \;[v_n(\alpha^*)+\sum_{k\neq n}b_k(\alpha^*)]-\sum_{k\neq n}b_k(\alpha^*_{-n}).
\end{aligned}
\end{equation}

\subsection{Desired Economic Auction Properties}\label{sec:Properties}
It is essential for an auction to have certain economic properties. First, we discuss these economic properties and then we prove that by using a VCG based auction approach these properties can be satisfied.

\begin{enumerate}
  \item Incentive Compatibility (truthfulness): An auction is incentive compatible if non of the bidders can get higher utility by not reporting its true evaluation vector. Based on this property, a dominant strategy for any bidder is to declare its true evaluation value regardless of what the other bidders do.
  \item Individual Rationality: An auction is individually rational if the utility $U_n$   for each bidder $n$ is greater or equal zero (i.e. $U_n\geq0$). Meaning that the winning bidders obtain non-negative utility (i.e. bidders do not pay more than their evaluation values) from the auction and no one suffer as a result of participating in the auction.
  \item No Positive Transfers: In auctions with no positive transfers, the payment of any bidder $n$ must be greater or equal zero (i.e. $p_n\geq0$). This prevents situations when the auctioneer has to pay the bidders.
\end{enumerate}

In Lemma \ref{compatibility}, Lemma \ref{rationality} and Lemma \ref{PositiveTransfer}, we show that by using the payment method and the VCG based auction approach discussed above, the aforementioned desired economic properties can be satisfied.
\begin{lem}
\label{compatibility}
Let $\mathbf{v_n}$ and $\mathbf{v_n^{'}}\neq \mathbf{v_n}$ be the $n^{th}$ bidder bidding vector when it is equivalent to its true evaluation values and any other values, respectively, and let $\alpha^*$ and ${\alpha^*}^{'}$ be the allocations that maximize the social welfare when $\mathbf{v_n}$ and $\mathbf{v_n^{'}}$ are declared, respectively. Then, for the $n^{th}$ bidder, the utility $U_n \geq {U^{'}}_n$.
\end{lem}
\begin {proof}
Using the utility definition and payment method in Section \ref{sec:Payment}, the utility of bidder $n$ is $U_n=v_n(\alpha^*)+\sum_{k\neq n}v_k(\alpha^*)-\sum_{k\neq n}v_k({\alpha}^*_{-n})$ when declaring $\mathbf{v_n}$ whereas the utility of bidder $n$ is ${U^{'}}_n={v}_n({\alpha^*}^{'})+\sum_{k\neq n}{v}_k({\alpha^*}^{'})-\sum_{k\neq n}{v}_k({\alpha}^*_{-n})$ when declaring $\mathbf{v^{'}}_n$. Since $\alpha^*$ is the allocation that maximizes the social welfare, we have the following inequality:
\begin{equation} \label{eq:welfare}
\sum_n v_n(\alpha^*)\geq \sum_n v_n({\alpha^*}^{'}).
\end{equation}
Now, by subtracting the term $\sum_{k\neq n}{v}_k({\alpha}^*_{-n})$ from both sides of equation (\ref{eq:welfare}), we get $U_n\geq {U^{'}}_n$ which is the incentive compatibility property.
\end{proof}

\begin{lem}
\label{rationality}
Let $\alpha^*$ and ${\alpha^{*}}_{-n}$ be the allocations that maximize the social welfare with and without node $n$'s participation, respectively, with the assumption that each bidder submits its true evaluation values. Then each bidder $n$ do not suffer as a result of participating in the auction and the auction's winners do not pay more than their evaluation values (i.e. $U_n\geq 0$).
\end{lem}
\begin {proof}
To show individual rationality, consider the utility of node $n$:
\begin{align*}
U_n =
&  \;v_n(\alpha^*)+\sum_{k\neq n}v_k(\alpha^*)-\sum_{k\neq n}v_k(\alpha^*_{-n})\\
\geq
&  \;\sum_{j}v_j(\alpha^*)-\sum_{j}v_j(\alpha^*_{-n})\\
\geq
&  \;0.
\end{align*}
The first inequality holds since $v_n(\alpha^*)+\sum_{k\neq n}v_k(\alpha^*)=\sum_{j}v_j(\alpha^*)$, $\sum_{j}v_j(\alpha^*_{-n})\geq\sum_{k\neq n}v_k(\alpha^*_{-n})$ and $\sum_{j}v_j(\alpha^*_{-n})\geq 0$. The second inequality holds because $\alpha^*$ is the allocation that maximizes the social welfare,  $\sum_{j}v_j(\alpha^*)$.
\end{proof}

\begin{lem}
\label{PositiveTransfer}
As a result of using the payment method in Section \ref{sec:Payment}, the auction has no positive transfers (i.e. $p_n\geq 0$ for each bidder $n$).
\end{lem}
\begin{proof}
From equation (\ref{eq:payment}), we have $p_n =\sum_{k\neq n}b_k(\alpha^*_{-n})-\sum_{k\neq n}b_k(\alpha^*) \geq 0$, since $\alpha^*_{-n}$ is the allocation that maximizes the social welfare without the $n^{th}$ bidder participation, $\sum_{k\neq n}b_k(\alpha^*_{-n})$.
\end{proof}
\subsection{Design Challenges} \label{DesignChallenges}
Truthfullness is one of the important properties that needs to be taking into consideration when designing a spectrum auction. Sealed secondary price auction and VCG auction are very preferable as they guarantee that bidders submit their true evaluation values. As mentioned before, VCG auction has many properties that are essential to have in a spectrum auction. However, VCG requires finding an optimal allocation which is NP-complete because of the spectrum spatial reusability property. In addition, VCG is vulnerable to frauds of the auctioneer and bid-rigging between the insincere auctioneer and greedy bidders \cite{Themis}. Therefore, VCG auction can not be used in a spectrum auction without countermeasures for fraud and bid-rigging. 

Bid-rigging between a greedy bidder and an auctioneer can occur for the benefit of both. Since the auctioneer is aware of all bidders' bidding values, he can collude with a greedy bidder and reveal the winning bid value to him. In Figure \ref{fig:challenges}, we show an example of a spectrum auction where the auctioneer auctions one frequency band $\mathcal{|M|}=1$ to four BSs (i.e. subnet $2$ of the frequency conflict graph that is shown in Figure \ref{fig:subnets}). The auctioneer runs a VCG auction that is equivalent to a sealed secondary price auction for one frequency band auction. In Figure \ref{fig:Bid-rigging}, we show an example of bid-rigging. Bidder $4$ is the winner and bidder $2$ is a greedy bidder who colludes with the auctioneer and learns about the highest bid. As a result, bidder $2$ bids a value that is higher than his true evaluation but a little bit less than the highest bid. By doing so, the auctioneer considers the bidding value of bidder $2$ to be the charging price for the winner (i.e.
 bidder $4$). By such a bid-rigging action, the auctioneer can make more profit and share the spoils with bidder $2$.

On the other hand, a fraud occurs when an insincere auctioneer overcharges the winner in order to increase his own profit. This results is an unexpected bad utility for the winner. In Figure \ref{fig:fraud}, bidder $4$ is the winner and the charging price should be $7$ which is equivalent to the second highest bid. However, the insincere auctioneer charges bidder $4$ at $7.9$ to obtain higher revenue. This is possible since all bidding values are sealed and bidders do not know about the bidding values of each other.

To avoid possible bid-rigging and frauds, a successful spectrum auction design needs to take into consideration securing the auction by making the auctioneer able to decide how to allocate the frequency bands while keeping the bidders actual bidding values unknown to the auctioneer. This is essential to avoid possible back-room dealing and ensure a secure spectrum auction.
\begin{figure}[h]
  \centering
  \subfigure[Bid-rigging between an insincere auctioneer and a greedy bidder]{%
  \label{fig:Bid-rigging}
  \includegraphics[height=2in, width=3.5in]{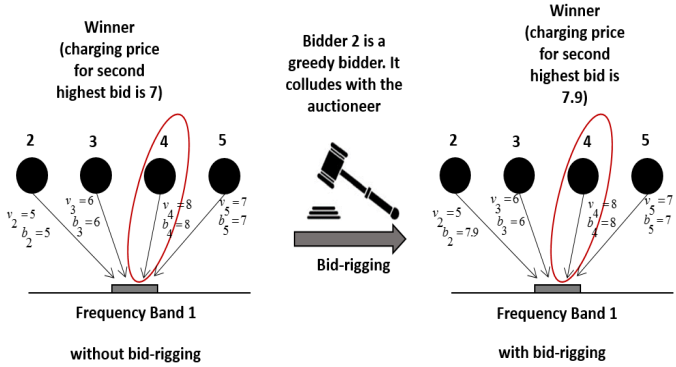}
  }\\%
\subfigure[Frauds of an insincere auctioneer]{%
  \label{fig:fraud}
  \includegraphics[height=1.4in, width=3in]{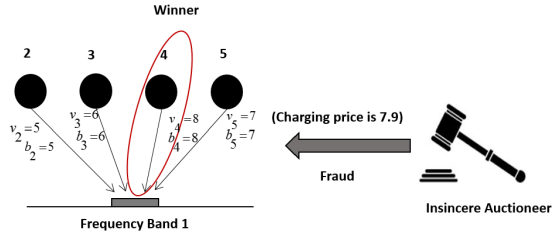}
  }%
\caption{Examples of bid-rigging and frauds in an unsecured spectrum auction of one frequency band and four BSs.}
\label{fig:challenges}
\end{figure}
\section{\textbf{MTSSA}: Secure Spectrum Auction Design} \label{sec:MTSSA}

In order to enable an efficient usage of the under-utilized shared spectrum managed by a broker. It is important to design a secure spectrum auction that allows the broker to provide a sufficient level of security, privacy and obfuscation to enable a reliable and efficient usage of the shared spectrum. In order to thwart back-room dealing, it is essential to have a mechanism that allows the auctioneer to find the maximum bid among all bidders without knowing their actual bids. The proposed \textbf{MTSSA} leverages Paillier cryptosystem to avoid possible frauds and bid-rigging.
In this section, we first describe Paillier cryptosystem and point out its special features. We then discuss \textbf{MTSSA} frequency bands allocation procedure. Finally, we present the security part of \textbf{MTSSA}.
\subsection{Paillier Cryptosystem}
Some of Paillier cryptosystem properties are essential for our secure spectrum auction design. Paillier cryptosystem \cite{Paillier:1999:PKC,Paillier:1999:EPC, Themis} is a probabilistic public key encryption system, i.e. the term probabilistic encryption indicates that when encrypting the same plaintext for multiple times it yields different ciphertexts, that satisfies special features such as homomorphic addition, indistinguishability and self blinding.

The homomorphic properties of Paillier cryptosystem provide it with a notable feature. As the encryption function of a message $m$, is given by $C(m)$, is additively homomorphic. i.e. $C(m_1+m_2)=C(m_1)C(m_2)$.
On the other hand, with the indistinguishability property of Paillier cryptosystem, if the plaintext $m$ is encrypted twice, the two created cyphertexts are different from each other and no one can distinguish the original plaintexts, except by random guessing, unless decrypting the original ciphertexts.
The self blinding property allows changing the ciphertext publicly without affecting the plaintext. Therefore, from the ciphertext $C(m)$, it is possible to compute a different randomized ciphertext $C^{'}(m)$ without knowing the decryption key or the original plaintext.

\subsection{Frequency Bands Allocation Procedure}\label{sec:Allocation}
All BSs that are interested in the auction and belong to the WSPs within the geographical region of the auction submit their bidding values to participate in the auction. Based on the location of these BSs  and which WSPs they belong to, the auctioneer creates an interference conflict graph (i.e. like the one in Figure \ref{fig:subnets}). The auctioneer executes the auction in one subnet at a time. For each subnet, the auctioneer selects a random BS $n\in \mathcal{N}$ to be the current root BS and considers its corresponding subnet, i.e. connected nodes/BSs. After solving for the current subnet, the auctioneer selects a new BS, that has not been a root BS before, to be the new root BS and excludes any previous root BS from its subnet along with the allocated frequency bands to these BSs. Following the same procedure, the auctioneer continues to execute the auction in one subnet after another until each BS has participated in the auction. Based on the subnet auction results, the auctioneer allocates the 
corressonding root node the frequency bands and charges it for the allocated resources. Each Bidder that is participating in the auction maintains a local conflict-table, i.e. as in \cite{multi-winner}. The bidder updates his bidding values if any BS within his interference range (i.e. connected to him in the interference conflict graph) wins frequency bands.

The \textbf{MTSSA} procedure is presented in the following steps:

$\mathbf{1.}$ \textbf{Each BS $n\in \mathcal{N}$ sets up its conflict-table and submits its encrypted version of bidding values $\mathbf{b_n}$:} Each WSP $l$ within the auctioneer's geographical region creates a set of all BSs $\mathcal{N}^l$ that are interested in bidding for the auctioneer's under-utilized frequency bands. Each BS $n\in \mathcal{N}^l$ creates a conflict-table with all the interfering BSs denoted by $\mathcal{I}_n$ (i.e. $\mathcal{I}_n$ is a set of all BSs that are located within the coverage area of BS $n$). Each WSP $l$ provides an update to each BS $n\in \mathcal{N}^l$ regarding its interfering  BSs $\mathcal{I}_n$\footnote{It is assumed that each WSP $l$ is aware of all BSs in its coverage area within the auction's geographical region, whether they belong to it or to other WSPs. Therefore, the set of interfering BSs $\mathcal{I}_n$ includes all BSs within the coverage area of BS $n$ that belong to WSP $l$ as well as BSs that belong to other WSPs.}. Each $n\in \mathcal{N}^l$ creates its 
bidding vector $\mathbf{b_n}$ that will be an input to a buffer that encrypts the bidding values. The encrypted bids are then submitted to federal gateway for randomization, see Section \ref{sec:MTSSA-Secure}, then sent to auctioneer, see Figure \ref{fig:Three-Stage-Auction}. Neither the auctioneer nor the other bidders know the actual bidding values $\mathbf{b_n}$ that BS $n$ has submitted. We show in Section \ref{sec:PaillierBidding} the procedure of encrypting the bidding values using Paillier encryption.

$\mathbf{2.}$ \textbf{Start with a random BS $n\in \mathcal{N}$ and consider its corresponding subnet:} The auctioneer does not have an optimal choice regarding which subnet it starts the auction from in order to maximize his revenue. Therefore, the auctioneer selects a random bidder $n$ from the set $\mathcal{N}$ and considers its corresponding subnet (i.e. a subnet consists of a root BS $n$ and all other BSs connected to BS $n$ in the interference conflict graph except BSs that have been previously considered root BSs).

$\mathbf{3.}$ \textbf{The auctioneer carries out a secure spectrum auction in the subnet of the selected BS $n$:} The auctioneer performs a secure spectrum auction procedure (detailed in Section \ref{sec:MTSSA-Secure}) in the current subnet under consideration.

$\mathbf{4.}$ \textbf{Allocate frequency bands and charge price:} Based on the subnet auction's results, the auctioneer allocates the root BS frequency bands and charges it for the allocated resources. The allocation and the payment vary based on the location of root BS and its relative bid with respect to neighboring BSs. Each winning BS stores its allocated frequency bands and its charging price in the conflict-table. 

$\mathbf{5.}$ \textbf{Proceed to next root BS:} A new root BS is selected based on a random selection done by auctioneer and the corresponding subnet secure bids are sent to auctioneer and the process is repeated starting from step $\mathbf{2}$.

Algorithm \ref{alg:MTSSA-Alg} summarizes $\textbf{MTSSA}$ spectrum auction procedure.

\begin{algorithm}[]
\caption{\textbf{MTSSA} Frequency Bands Allocation}\label{alg:MTSSA-Alg}
\begin{algorithmic}
	\STATE {$\mathcal{N}=\mathcal{N}^1 \cup \mathcal{N}^2...\cup \mathcal{N}^L$} 
	\COMMENT {i.e. $\mathcal{N}$ is the set of all BSs in the interference conflict graph}
	\STATE {$\mathcal{N}_0=\mathcal{N}$} 
	\STATE {Auctioneer generates his private and public keys of Paillier cryptosystem} 
	\STATE {Auctioneer sends public key and element $x$ to all BSs via pilot channel} 
        \WHILE{$\mathcal{N}\; != \phi$}
		\STATE{$n=random(\mathcal{N})$} 
		\COMMENT{Auctioneer selects a random BS}
		\STATE{$\mathcal{N}_n = include\_conflict(n)$}
		\COMMENT{Auctioneer adds BSs that form the $n^{th}$ subnet from conflict-table of $n^{th}$ BS as root} 
		\STATE{$\mathcal{N}_{-} =(\mathcal{N}_0 \setminus \mathcal{N})\cap \mathcal{N}_n$}
		\COMMENT{$\mathcal{N}_{-}$ is set of previous root BSs in the $n^{th}$ subnet}
		\STATE{$\mathcal{N}_n = \mathcal{N}_n \setminus \mathcal{N}_{-} $}
		\COMMENT{Auctioneer removes from $\mathcal{N}_n$ previous root BSs}
		\STATE{$\mathcal{M}_{-} = include\_alloc(\mathcal{N}_{-})$}
		\COMMENT{$\mathcal{M}_{-}$ is set of freq. bands allocated to $\mathcal{N}_{-}$}
		\STATE{$\mathcal{M}_n = \mathcal{M} \setminus \mathcal{M}_{-} $}
		\COMMENT{Auctioneer removes from $\mathcal{M}$ freq. bands allocated to $\mathcal{N}_{-}$}
		\STATE{$\mathcal{K}_n = alloc\_vect(\mathcal{N}_n, \mathcal{M}_n)$}
		\COMMENT{Auctioneer forms allocation vector $\mathcal{K}_n$ and sends to $\mathcal{N}_n$} 
                \STATE{BSs $\in \mathcal{N}_n$ send encrypted bids to federal gateway}
                \STATE{Federal gateway randomizes bids and forward to auctioneer}
                \STATE{Auctioneer selects the highest allocation $\alpha^{\star}$}
                \STATE{Auctioneer charges price $p_n$ to BS $n$}
                \STATE{$\mathcal{N}=\mathcal{N} \setminus \{n\}$}
       \ENDWHILE
\end{algorithmic}
\end{algorithm}

\subsection{Secure Spectrum Auction Using Paillier Cryptosystem}\label{sec:MTSSA-Secure}
In order for {\textbf{MTSSA} to ensure a secure auction, it is important to design {\textbf{MTSSA} such that no way for the auctioneer to manipulate the auction. VCG auction is proven to have the incentive compatibility property from the bidders side which is essential for our design. In order for {\textbf{MTSSA} to prevent the auctioneer from conducting any frauds or bid-rigging \cite{Themis}, it is important to limit the auctioneer's capability by making him only able to exploit the winners and their payments without knowing the actual bidding values. So, by leveraging Paillier cryptosystem in our design, {\textbf{MTSSA} can ensure a secure spectrum auction. Next we discuss in details how both the bidders and the auctioneer need to collaborate in order to carry out a secure spectrum auction.

\subsubsection{Impact of Paillier Cryptosystem on the Bidding Values}\label{sec:PaillierBidding}

\indent \textbf{Encrypting the Bidding Values:} Each bidder submits its bidding values through a buffer that uses Paillier cyptosystem to encrypt the bidding values and create a vector of ciphertexts for each bidding value. Let $s$ be a number that any actual bidding value does not exceed and let $z=b(\alpha)$ be the actual bidding value for allocation $\alpha$ such that $1\leq z\leq s$. Let the vector of ciphertexts for $z$ be $\mathbf{c}(z)$ that is given by
\begin{equation} \label{eq:siphertexts}
\begin{aligned}
\textbf{c}(z)=
&  \;(c^1,...,c^s) \\
=
&  \;(\underbrace{C(x),...,C(x)}_z,\underbrace{C(0),...,C(0)}_{s-z}),
\end{aligned}
\end{equation}
where $C(x)$ is the Paillier encryption of the public element $x$ (i.e. $x\neq 0$) and $C(0)$ is the Paillier encryption of $0$. As mentioned before, $C$ has a self blinding property which makes $z$ undeterminable without decrypting the elements in  $\textbf{c}(z)$.\\

\textbf{Selecting the Maximum Bidding Value:} The auctioneer can determine the bidder with the maximum bidding value from the encrypted bidding values without knowing their actual values. Let $\mathbf{c}(z_j)=(c_j^1,...,c_j^s)$ be the encrypted bidding vector of bidder $n$ for certain allocation $\alpha$. First, consider the product of all encrypted bidding vectors for allocation $\alpha$,
\begin{equation}
\prod_{j} \mathbf{c}(z_j)= (Q_1,...,Q_s) = (\prod_{j} c_j^1,...,\prod_{j} c_j^s).
\end{equation}

Due to the homomorphic addition property of Paillier cryptosystem, $Q_i$ (i.e. $1 \leq i \leq s$ and $i \neq j$) is equivalent to
\begin{equation}
Q_i=\prod_{j} c_j^i = C^{\gamma(i)}(x) = C(\gamma(i)x),
\end{equation}
where $\gamma(i)$ represents the number of values that are equal to or greater than $i$, i.e. $\gamma(i)=|\{j : z_j \geq i\}|$. Given that $\gamma(i)$ monotonically decreases when $i$ increases, one way to find the maximum of these bidding values is to decrypt $Q_i$ and check whether the decrypted value $C^{-1}(Q_i)$ equals $0$ or not for $i$ changing from $s \rightarrow 1$. Once the largest $i$ with a decrypted value $C^{-1}(Q_i) \neq 0$ is found, then the maximum bidding value for the allocation $\alpha$ is determined to be $i$ (i.e. $i=\max \{z_j\}$).\\

\textbf{Randomizing the Encrypted bidding Values:} Without knowing $z$, the federal gateway adds a constant $t$ to the encrypted vector $\mathbf{c}(z)$ and randomizes the rest of its elements. This results in the following vector
\begin{equation}
\mathbf{c^{'}}(z+t)=(\underbrace{C(x),...,C(x)}_z,c_1^{'},...,c_{s-z}^{'}),
\end{equation}
where $t$ can not be obtained from either $\mathbf{c}(z)$ or $\mathbf{c^{'}}(z+t)$ because of the self blinding property of Paillier cryptosystem. In addition, $t$ can not be figured out by comparing $\mathbf{c}(z)$ and $\mathbf{c}(z+t)$ during the shifting and randomizing process.

\subsubsection{Securing the \textbf{MTSSA} Subnet Auction}
By using Paillier cryptosystem as discussed in Section \ref{sec:PaillierBidding}, with encrypted bidding values it is still possible to find the maximum bid and the encrypted bidding vectors are randomized without knowing their actual values. This makes it possible to apply a VCG based auction in each subnet. As mentioned before, the proposed \textbf{MTSSA} auction is carried out in one subnet after another. In certain subnet, \textbf{MTSSA} auction  is performed as follows:

$\mathbf{1.}$ The auctioneer generates his private and public keys of Paillier cryptosystem and publishes his public key and element $x$ over the pilot channel.

$\mathbf{2.}$ Each BS submits its sealed bidding vector $\mathbf{b_n}=\mathbf{v_n}$ (i.e. its true evaluation values since we consider a VCG based auction). The auctioneer creates representing vectors $\mathbf{C}_T=\mathbf{C}(O)$, $\mathbf{C}_1=\mathbf{C}(O)$,...,$\mathbf{C}_N=\mathbf{C}(O)$ where $N$ is the number of bidders (i.e. BSs), the initial $O(\alpha)$ equals $0$ and the size of vector $\mathbf{C}$ equals $|\mathcal{K}|$. In order for bidder $z$ to keep his bidding value $b_z$ secret, he adds his encrypted bidding value to all of the representing vectors except $\mathbf{C}_z$. Once all bidders are done performing this addition, the auctioneer obtains
\begin{equation}\label{eq:representing-E}
\mathbf{C}_T=(\prod_{n} \mathbf{c}(b_n(\alpha_1)),...,\prod_{n} \mathbf{c}(b_n(\alpha_{|\mathcal{K}|}))).
\end{equation}

Due to the homomorphic addition property of Paillier cryptosystem, equation (\ref{eq:representing-E}) is equivalent to
\begin{equation}\label{eq:representing-E-homo}
\begin{aligned}
\mathbf{C}_T= (\mathbf{c}(\sum_{n} b_n(\alpha_1)),...,\mathbf{c}(\sum_{n} b_n(\alpha_{|\mathcal{K}|}))) = \mathbf{C}(\sum_{z} b_z),
\end{aligned}
\end{equation}
and
\begin{equation}
\mathbf{C}_z=\mathbf{C}(\sum_{n\neq z} b_z) \;\;\; 1\leq z \leq N.
\end{equation}

$\mathbf{3.}$ Federal gateway adds $\theta(\alpha)=t$ to $\mathbf{C}_T, \mathbf{C}_1, ..., \mathbf{C}_N$ to obtain $\mathbf{C}(\sum_{n} b_n+\theta)$ and $\mathbf{C}(\sum_{n \neq z} b_n+\theta) \forall z$. It sends these randomized encrypted bids to auctioneer.\\

$\mathbf{4.}$ In order for the auctioneer to select an allocation for the current subnet and find its corresponding charging price, it finds the maximum sum value
\begin{equation}\label{eq:maxElement}
\begin{aligned}
g=
& \;\arg \max_{\alpha \in \mathcal{K}}(\sum_{n} b_n(\alpha)+\theta(\alpha)) \\
=
& \;\arg \max_{\alpha \in \mathcal{K}}(\sum_{n} b_n(\alpha))+t,
\end{aligned}
\end{equation}
which can be determined by the auctioneer by taking the product of all the encrypted elements in $\mathbf{C}_T$ , i.e. as discussed in Section \ref{sec:PaillierBidding}, which is equivalent to $\prod_{i=1}^{|\mathcal{K}|} \mathbf{c}(\sum_{n=1}^N b_n(\alpha_i)+t)$. The auctioneer determines the maximum element in that product which is equivalent to $g$ in (\ref{eq:maxElement}).

$\mathbf{5.}$ For each allocation $\alpha$, the auctioneer decrypts the $g^{\text{th}}$ element of vector $\mathbf{c}(\sum_{n} b_n(\alpha)+\theta(\alpha))$ in $\mathbf{C}_T$ and finds whether it equals $0$ or $x$. If it equals $x$ at allocation $\alpha^{*}$, then the auctioneer selects $\alpha^{*}$ to be the allocation that maximizes $\sum_n b_n$ in the current subnet and considers its corresponding BSs.

$\mathbf{6.}$ To find the charging price for the root BS, the auctioneer decrypts $\mathbf{c}(\sum_{n\neq z} b_n(\alpha^{*})+\theta)$ of $\mathbf{C}_z$ and finds the masked value $(\sum_{n\neq z} b_n(\alpha^{*})+ \theta)$.

$\mathbf{7.}$ The auctioneer then finds the maximum masked bid of the product of the encrypted elements $\max_{\alpha \in \mathcal{K}}(\sum_{n\neq z} b_n(\alpha)+t)$, similar to Step $\mathbf{4}$, which is equivalent to $(\sum_{n\neq z} b_n(\alpha_{-z}^{*})+t)$.

$\mathbf{8.}$ The auctioneer then finds the charging price for root BS of allocation $\alpha^{*}$ that is given by
\begin{equation}
p_z= (\sum_{n\neq z} b_n(\alpha_{-z}^{*})+t)-(\sum_{n\neq z}b_n(\alpha^{*})+t).
\end{equation}

\section{Simulation and Analysis}\label{sec:Simulation}
In this section, we first evaluate the performance of the proposed \textbf{MTSSA} spectrum auction and compare it with the performance of other spectrum auction mechanisms. Three performance metrics are considered: spectrum utilization, auctioneer's revenue and bidders' satisfaction. These are the most important performance metrics that need to be maximized in a successful spectrum auction. In addition, we analyze the security strategy of the proposed secure spectrum auction \textbf{MTSSA} that makes it able to avoid possible frauds and bid-rigging.
\subsection{Performance Analysis}
We consider a spectrum auction hosted by the auctioneer (the broker) in a $A*A$ $m^2$  square geographical region with two cellular networks located within the same region where the auction takes place. Each cellular network belongs to different WSP, i.e. there exists two WSPs $L=2$ that are interested in participating in the spectrum auction. Each WSP has certain number of BSs, located in macro cells or small cells, that are interested in bidding for the auctioneer's under-utilized frequency bands. The BSs are randomly placed in the auction's geographical area. Suppose that the frequency mutual interference between any two BSs is based on the distance between them. Any two macro cells' BSs located within a distance of $0.4A$ can not be allocated the same frequency bands and these BSs are connected together in the frequency conflict graph. Also, any small cell's BS can not be allocated the same frequency bands of any other BS located within a distance of $0.05A$ from it. In our simulation setup, bids are 
selected randomly with biding per frequency band is monotonically decreasing, i.e the BS's bid for first frequency band is higher than second frequency band and second frequency band  is higher than third frequency band  and so on, see \cite{Ahmed_Utility1, Ahmed_Utility2, CarrierAggregation}.

Based on the frequency assignment policy, three spectrum auction mechanisms are considered in our simulation as described in the following three cases:
\begin{itemize}
\item $Case\;1:$ Conventional spectrum leasing (\textbf{CSL}) case where the government directly leases the under-utilized spectrum to each WSP with heist bid. Once the WSP is assigned certain frequency bands, it then allocates these resources internally to its BSs.
\item $Case\;2:$ \textbf{MTSSA} where each WSP directly submits all of its BSs' encrypted bids to the auctioneer. The auctioneer decides the frequency bands allocation to each BS whereas the WSP has no control on the resources allocation process. By using \textbf{MTSSA} frequency assignment process, each BS can be allocated any number of frequency bands between zero and all of the auctioneer's under-utilized frequency bands.
\item $Case\;3:$ \textbf{MTSSA} with fixed limit (\textbf{MTSSA-FL}) is a special case of the proposed \textbf{MTSSA} where the frequency bands allocation policy is similar to the proposed \textbf{MTSSA} but is restricted in the number of frequency bands that each BS can bid for. Each BS bids for a fixed number of frequency bands and it can be allocated any number of frequency bands between zero and that fixed number of frequency bands it submitted the bids for.
\end{itemize}

We ran Monte Carlo Simulation, for the three cases described above, and the results are averaged over $25$ independent runs in which the location and the bidding values of the BSs are generated randomly and the performance metrics are evaluated. We consider the network setup described above with different number of macro cells and small cells' BSs that belong to the two WSPs. First, we consider $8$ BSs, i.e. $4$ macro cells' BSs and $4$ small cells' BSs. Second, we consider $12$ BSs, i.e. $6$ macro cells' BSs and $6$ small cells' BSs. Third, we consider $16$ BSs, i.e. $8$ macro cells' BSs and $8$ small cells' BSs.

We consider three performance metrics to compare between \textbf{CSL}, \textbf{MTSSA}, and \textbf{MTSSA-FL}. These performance metrics are:
\begin{itemize}
\item Spectrum Utilization: It is represented by the sum of the frequency bands that are allocated by the auctioneer to the winning BSs.
\item Auctioneer's Revenue: It is given by the sum of all BSs' payments, i.e $R=\sum_{n=1}^{n=N}p_n$.
\item Bidders' Satisfaction: It is represented by the sum of all winning BSs' utilities divided by the sum of all BSs' evaluation values, i.e. $\sum_{n\in \mathcal{A}}U_n/\sum_{n\in \mathcal{N}}v_n$, where $\mathcal{A}$ is the set of BSs that are allocated frequency bands.
\end{itemize}

In Figure \ref{fig:Performace}, we compare the performance of the proposed \textbf{MTSSA} and its special case \textbf{MTSSA-FL} with that of a \textbf{CSL} based auction. We plot the spectrum utilization, auctioneer’s revenue and BSs’ satisfaction of the three auction designs with different number of BSs, i.e. 8 BSs, 12 BSs and 16 BSs as mentioned before.

Figure \ref{fig:SpectrumUtilization} shows the spectrum utilization versus the number of available under-utilized frequency bands. As the number of frequency bands increases, the spectrum utilization, which is represented by the number of the allocated frequency bands, also increases for each of the three auction mechanisms. For certain number of frequency bands, each of the three mechanisms shows higher utilization when the number of BSs (bidders) increases. However, it is not surprising that the performance in terms of utilization for \textbf{CSL} is lower than that for the other two mechanisms. This is because in \textbf{CSL}, the auctioneer assigns each WSP different frequency bands and the frequency bands assigned to each WSP are then auctioned among BSs that belong to that WSP. In the case of \textbf{CSL}, the auctioneer considers one frequency conflict graph for each WSP and frequency reusability is not applicable among BSs that belong to different WSPs, i.e. BSs that belong to different WSPs and are 
not within the interference range of each other are not allowed to use the same frequency bands. Moreover, the utilization in the cases of \textbf{MTSSA} and \textbf{MTSSA-FL} is almost the same when the number of available frequency bands is low and is slightly higher for \textbf{MTSSA} than that for \textbf{MTSSA-FL} when the number of available frequency bands is higher.

Figure \ref{fig:Revenue} shows that for each of the three mechanisms, the auctioneer's revenue increases when the number of BSs increases. This is expected as the auctioneer's revenue increases with more bidders requesting more resources. However, for certain number of BSs, the auctioneer's revenue for \textbf{MTSSA-FL} is higher than that for \textbf{MTSSA} and \textbf{CSL} and as expected the auctioneer's revenue is the lowest in the case of \textbf{CSL}. The bump of \textbf{MTSSA} over \textbf{CSL} is from the payments received from winning BSs that belong to different WSPs, and not located within the interference range of each other, but are allocated similar frequency bands.

We show in Figure \ref{fig:BiddersSatisfaction} that as the number of frequency bands increases, the bidders' satisfaction also increases until it saturates when each bidder is allocated the number of frequency bands he bids for. On the other hand, the bidders’ satisfaction for \textbf{CSL} is higher than that for \textbf{MTSSA} and \textbf{MTSSA-FL} due to the less number of BSs competing for resources as all the frequency bands are allocated to one WSP. Therefore, the bidders who belong to the winning WSP get their requested frequency bands while paying less.

The comparison between the three mechanisms in Figure \ref{fig:Performace} shows that \textbf{MTSSA} behaves well in terms of performance and proves to be better than the conventional spectrum leasing mechanism as it considers spectrum reusability and in the same time it guarantees a secure spectrum auction.

\begin{figure}[h!]
  \centering
\subfigure[Spectrum Utilization]{%
  \label{fig:SpectrumUtilization}
  \includegraphics[height=2.5in, width=3.5in]{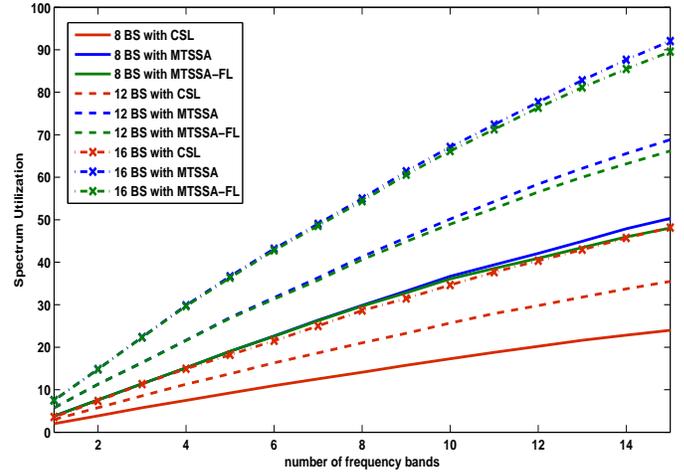}
  }\\%
  \subfigure[Auctioneer's Revenue]{%
  \label{fig:Revenue}
  \includegraphics[height=2.5in, width=3.5in]{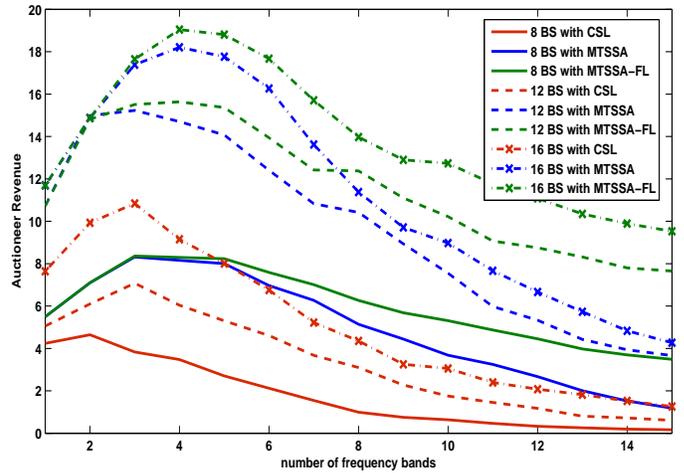}
  }\\%
\subfigure[Bidders' Satisfaction]{%
  \label{fig:BiddersSatisfaction}
  \includegraphics[height=2.5in, width=3.5in]{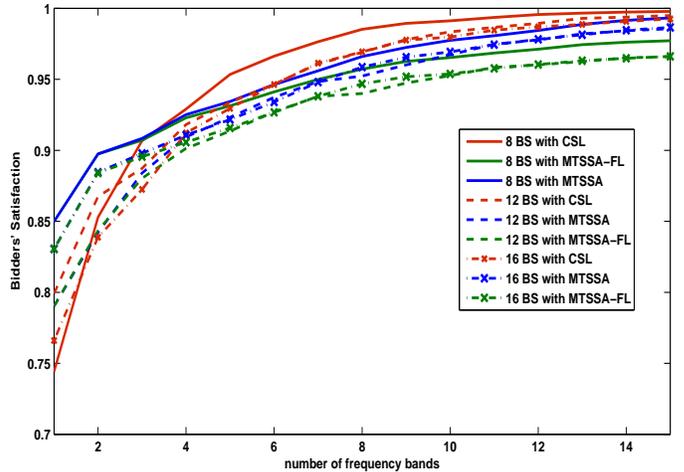}
  }%
\caption{Performance comparison of \textbf{MTSSA}, \textbf{MTSSA-FL} and \textbf{CSL}.}
\label{fig:Performace}
\end{figure}

\subsection{\textbf{MTSSA} Security Analysis}
As discussed before, our proposed \textbf{MTSSA} leverages Paillier cryptosystem in order to ensure that the BSs' bidding values are kept unknown to the auctioneer while the auctioneer is still able to find the winners and charges them their corresponding payments. This is possible because of the indistinguishability property of Paillier cryptosystem, i.e. it is not possible to know the value of $z$ without decrypting each element in $c(z)$, and the self blinding property that makes it impossible to find a mapping function from $c(t)$ to $c^{-1}(z+t)$ \cite{Paillier:1999:PKC,Paillier:1999:EPC, Themis}. In order to prevent an insincere auctioneer from performing any frauds, we consider a secure gateway that is operated by federal government. Its main function is to randomize the encrypted bids by the random constant $t$. So auctioneer can determine the winning allocation and assign secondary price without any knowledge of the original bidding values of BSs. This way \textbf{MTSSA} can avoid bid-rigging between an insincere auctioneer and a greedy bidder. This can 
be guaranteed because even if certain bidder colludes with auctioneer, he can not find out about the bidding values as federal gateway randomized it. Therefore, all BSs that belong to different WSPs are treated equally by the proposed \textbf{MTSSA} and their bidding values are kept secret from the auctioneer who is only able to determine the winners and their corresponding charged price.


\section{Conclusion}\label{sec:conclude}
In this paper, we propose a secure spectrum auction \textbf{MTSSA} for a multi-tier dynamic spectrum sharing system. By considering the spectrum reusability property, \textbf{MTSSA} enables an efficient sharing of the under-utilized frequency bands with commercial WSPs. In order to allow spectrum reuse among multiple WSPs, the frequency conflict graph that is considered by \textbf{MTSSA} includes all BSs that belong to multiple WSPs and the auction is carried out in one subnet after another. \textbf{MTSSA} leverages Paillier cryptosystem and a federal gateway to keep the BSs' bidding values unknown to the auctioneer. The auctioneer uses the additive homomorphic property of Paillier cryptosystem to find the winning bidders and their charging prices. This prevents possible frauds and bid-rigging between an insincere auctioneer and greedy bidders. Compared with conventional spectrum leasing mechanism, \textbf{MTSSA} has shown better performance while providing a secure spectrum auction against possible back 
room dealings.

\bibliographystyle{ieeetr}
\bibliography{pubs}
\end{document}